\documentclass[sts]{imsart}

\RequirePackage{graphicx}
\RequirePackage{amsthm,amsmath,amsfonts,amssymb}
\usepackage{mathtools}
\usepackage{multirow}
\usepackage{enumitem}
\usepackage{url}
\usepackage{bm}
\usepackage{colortbl}
\usepackage{xcolor,cancel}

\definecolor{Gray}{gray}{0.9}
\RequirePackage[authoryear]{natbib}
\RequirePackage[colorlinks,citecolor=blue,urlcolor=blue]{hyperref}

\newcommand{\argmin}{\mathrm{argmin}}

\newlength{\dhatheight}

\newtheorem{theorem}{Theorem}[section]

\newtheorem{proposition}[theorem]{Proposition}
\startlocaldefs

\endlocaldefs

\usepackage[defaultcolor=purple]{changes} 

\begin{document}

\begin{frontmatter}
\title{A Bayesian `sandwich' for variance estimation}

\begin{aug}
\author[A]{\fnms{Kendrick}~\snm{Li}\ead[label=e1]{kendrick.li@stjude.org}}\and
\author[B]{\fnms{Kenneth}~\snm{Rice}\ead[label=e2]{kenrice@uw.edu}}

\address[A]{Department of Biostatistics, St. Jude Children's Research Hospital, Memphis, Tennessee, U.S.A.\printead[presep={\ }]{e1}.}
\address[B]{Department of Biostatistics, University of Washington, 
Seattle, Washington, U.S.A.\printead[presep={\ }]{e2}.}
\end{aug}

\begin{abstract}
Large-sample Bayesian analogs exist for many frequentist methods, but are less well-known for the widely-used `sandwich' or `robust' variance estimates. We review existing approaches to Bayesian analogs of sandwich variance estimates and propose a new  analog, as the Bayes rule under a form of balanced loss function, that combines elements of standard parametric inference with fidelity of the data to the model. Our development is general, for essentially any regression setting with independent outcomes. Being the large-sample equivalent of its frequentist counterpart, we show by simulation that Bayesian robust standard error estimates can faithfully quantify the variability of parameter estimates even under model misspecification -- thus retaining the major attraction of the original frequentist version. We demonstrate our Bayesian analog of standard error estimates when studying the association between age and systolic blood pressure in NHANES.
\end{abstract}


\end{frontmatter}
\section{Introduction}

Probabilistic models, with a finite number of parameters that characterizes the data generating distribution, are a standard statistical tool. Frequentist methods view the parameters as fixed quantities and use statements about replicate datasets to make inference, while the Bayesian paradigm uses the `language' of probability to directly describe uncertainty about parameters. This philosophical distinction can lead to differences in interpretation of analyses, depending on which framework is being used.

Despite this, in analysis of large samples standard methods both approaches tend to have close analogs under the other. For example, frequentist maximum likelihood estimates (MLEs) and Bayesian posterior means behave similarly in large samples, as do typical standard error estimates and posterior standard deviations --- fundamentally due to Bernstein-von Mises' theorem~\cite[\S 10.2]{van2000asymptotic}. 

Complications arise, however, when we consider model violations.

\subsection{Literature review: Frequentist and Bayesian statistics under model violations}

Frequentist inference in large samples remains operationally straightforward; under only mild regularity conditions straightforward `robust' variance estimates~\citep{royall1986model} provide inference for the parameters consistently estimated by the MLE or other $M$-estimators, even under model mis-specification. These robust methods (also known as `sandwich' estimates of variance) are widely-used, extremely simple to calculate, and are available in standard software. Interpreting the parameter that is being estimated can be more of a challenge. For example, without assuming that a regression model underpins the data-generating process, we may nevertheless usefully view some regression coefficients as summaries of trends over the whole population~\citep{buja2019modelsi}. More generally -- and more abstractly -- the parameter consistently estimated by the MLE of a mis-specified model is the form of that model which is closest in Kullback-Leibler divergence to the true data-generating distribution~\citep{huber1967behavior}. We shall refer to this parameter value as the `minimal Kullback-Leibler point', as known as the `pseudotrue parameter'~\citep{bunke1998asymptotic}.

For Bayesian approaches, the interpretation issues are similar; large sample results prove that under model violations the posterior mean approaches the minimal Kullback-Leibler point, subject to  mild regularity conditions, similarly to the MLE. Formally, \citet{berk1966limiting} and \citet{berk1970consistency} established the weak convergence of posterior distribution to a point mass at the pseudotrue parameter value. \citet{bunke1998asymptotic} and \citet{kleijn2012bernstein} investigated the usual Bayesian estimate (posterior mean) and formally established its consistency to the minimal Kullback-Leibler point and asymptotic Normality. Interpreting Bayesian point estimates seems no more challenging than those of MLEs or similar non-Bayesian estimates. Posterior variances are more complex: rather than just quantify the variability of estimates under repeated sampling (regardless of whether this sampling happens under the specified model) they quantify uncertainty of beliefs about the pseudotrue parameter where we have assumed that the potentially mis-specified model is actually true. As discussed carefully by \citet{walker2013bayesian}, reporting the standard model-based posterior variance is therefore hard to justify when model mis-specification is a concern.

To provide a Bayesian resolution of this concern, one approach relies on modifying the likelihood used, so that the corresponding posterior distribution's variance behaves similarly to the sandwich variance estimate. For linear regression with independent outcome with potentially heteroscedastic residuals, \citet{startz2012bayesian} proposed a Bayesian analog of robust variance estimation by flexibly modeling the covariance of the residual error and using the likelihood of just some moments of the dependent variables, rather than the whole vector of outcomes. Also focusing on linear regression but allowing the true data generating process to be nonlinear, \citet{szpiro2010model} gave a fully Bayesian analog of the sandwich covariance estimate, as the posterior variance of a specific linear summary of the parameter space, in the limiting situation of an extremely flexible model and prior -- essentially a form of `non-parametric Bayes'~\citep{dunson2010nonparametric}. These results, however, are to date not available for other regression estimates or more general models. Relaxing the requirement that a fully-specified model be used, \citet{kim2002limited} uses a limited information likelihood that pivots around the identifying estimating equations (rather than the data generating process) as an instrument such that the posterior distribution weakly converges to the asymptotic distribution of the generalized method of moments estimators, of which the asymptotic variance is the sandwich variance estimate. Similarly, \citet{muller2013risk} proposes pivoting inference on an artificial normal likelihood function centered at the MLE with the sandwich variance estimate as the variance-covariance matrix, such that the variance of the `sandwich posterior' distribution converges the sandwich variance; \citet{hoff2013bayesian} similarly uses the joint likelihood of the parameter of interest and the variance of the score function  to derive another `Bayesian sandwich posterior' with similar large sample properties.

Other methods for Bayesian robust inference exist which do not rely on the sandwich variance. Drawing inspiration from physical models, \citet{o2013bayesian} defends model-based inference but suggests including a nonparametric discrepancy term to account for deviation from that model. \citet{miller2018robust} constructed a likelihood that conditions only on the distribution of the data being close to the distribution expected under a model, with the definition of “close” determining the type of robustness obtained.  \citet{nott2023bayesian} reviewed several other likelihood-based approaches for Bayesian robust inference. A further alternative Bayesian approach to robustness is introduced by \citet{lancaster2003note}, who follows up the well-studied connection between the robustness of sandwich estimates and bootstrapped covariance estimates (see e.g. \citet{buja2019modelsi,buja2019modelsii}) by instead considering the Bayesian bootstrap~\citep{rubin1981bayesian}. This replaces the the regular boostrap’s discrete re-sampling with smoother (re)weighting, based on a Dirichlet prior, but retains the classic bootstrap's robustness properties. However, as \citet{rubin1981bayesian} notes when introducing the Bayesian bootstrap, it involves `somewhat peculiar” model assumptions, distinct from those typically made under any form of regression modeling.

\subsection{Manuscript outline}

Despite the expansive literature on this topic, we see that a general Bayesian analog of sandwich variance estimates derived as an extension of model-based arguments -- much as \citet{royall1986model} views them as an extension of classical modeling methods --  is yet to be developed. In this article, we propose such a Bayesian analog, using parametric models and priors for which we calculate the Bayes rule with respect to a form of \textit{balanced} loss function~\citep{zellner1994bayesian}. We introduce the balanced loss function and its theoretic properties in Section~\ref{sec:bil}. The balanced loss function comprises two terms, penalizing estimation error and lack of model fit respectively. In Section~\ref{sec:sim}, we demonstrate the proposed Bayesian robust variance estimate with several examples and use simulation studies to show its behavior in various settings. In Section~\ref{sec:app}, we use Bayesian robust variance estimates to quantify the uncertainty in an analysis to study the association between systolic blood pressure and age using a dataset from the 2017-2018 National Health and Nutrition Examination Survey (NHANES). We conclude with a discussion in Section~\ref{sec:discussion}. Proofs and additional discussion are deferred to the Supplementary Material.

\section{Inference Loss functions}

For general observations we denote  $\bm{\mathsf{{Z_{n}}}}=(\bm Z_1, \dots, \bm Z_n)$ as an independent and identically distributed random sample from a distribution $P_0$ with density $p_0$. Suppose an analyst models the data generating process by a parametric model $\mathcal P = \{P_\theta:\theta \in\Theta \subset \mathbb R^p\}$, the distributions in which are indexed by a real-valued parameter $\theta$. In our major focus on regression models, the $\bm Z_i$'s are the combined outcome and explanatory variables, i.e. $\bm Z_i = (Y_i, \bm X_i)$ where $Y_i\in\mathbb R$ is the outcome variable and $\bm X_i$ is a real vector of explanatory variables. To help orient readers to which statistical paradigm is being employed, for frequentist analyses we use $\theta$ to denote the unknown parameter, while for Bayesian analyses, we use $\vartheta$.We note that  parameters $\vartheta$ act as an instrument for the analysis that is associated with a prior distribution and a posterior distribution derived from a (possibly wrong) likelihood function.  We write $\pi(\cdot)$ as the density function of the corresponding prior distribution, supported on $\Theta$, and $\pi_n(\cdot\mid\bm{\mathsf{z_n}})$ and $\Pi_n(\cdot\mid \bm{\mathsf{z_n}})$ as the posterior density function and {distribution function} of $\vartheta$ conditional on $\bm{\mathsf{{Z_{n}}}}=\bm{\mathsf{z_n}}$, respectively. Here the subscript emphasizes $n$ that the posterior density and distribution function depends on the data. We denote $E_{P}$ and $Var_{P}$ as the expectation and variance (or variance-covariance matrix), respectively, with respect to a posterior measure $P$, such that $E_{P_0}(\cdot)$ and $Var_{P_0}(\cdot)$ respectively denote the expectation and variance (or variance-covariance matrix) with respect to the data generating process, and $E_{\Pi_n}(\cdot\mid \bm{\mathsf{z_n}})$ and $Var_{\Pi_n}(\cdot\mid \bm{\mathsf{z_n}})$ respectively denote the posterior expectation and posterior variance (or variance-covariance matrix). For brevity, we will use $E_0$ and $Var_0$, respectively, in place of $E_{P_0}$ and $Var_{P_0}$. Throughout the article, we will often write $E_{\Pi_n}(\cdot\mid \bm{\mathsf{Z_n}})$ and $Var_{\Pi_n}(\cdot\mid \bm{\mathsf{Z_n}})$ to emphasize the that the Bayesian posterior expectation and variance are random variables themselves, i.e. functions of the random data $\bm{\mathsf{Z_n}}$.

Suppose that $\theta\in\mathbb R^p$ is the parameter of interest and $d$ a corresponding estimate. Using decision theory, optimal Bayesian estimates, known as \textit{Bayes rules}, are obtained by minimizing the posterior risk $E_{\Pi_n}\{L(\vartheta, d)\mid\bm{\mathsf{Z_n}}\}$, where $L(\cdot, d)$ is a loss function describing the discrepancy between the decision  and the parameter. Common loss functions include the $L_1$-loss $\lVert \theta - d\rVert_1$ and $L_2$-loss $\lVert \theta - d\rVert_2$, for which the Bayes rules are the posterior median and posterior mean, respectively. For a more comprehensive review of decision theory and optimality of  Bayes rules, see~\citet{parmigiani2009decision}. 

To give rules that estimate $\theta$ but also indicate the precision of that estimate, we need more general loss functions. To achieve both those goals we consider what we shall call the \textit{inference loss function}
\begin{equation}\label{eq:inference-loss}
    L_I(\theta, d, \Sigma) = \log\lvert\Sigma\rvert+ (\theta - d)^T\Sigma^{-1}(\theta - d),
\end{equation}
where $\Sigma$ is a $p\times p$ positive definite matrix. The right-hand term is a sum of adaptively-weighted losses, each of the form $(\Sigma^{-1})_{ij}(\theta_i - d_i)(\theta_j - d_j)$, where the weights are determined by the elements of the matrix $\Sigma^{-1}$. To prevent the weight for any combination of the $(\theta_i - d_i)(\theta_j - d_j)$ reaching zero, we penalize by the left term, the log of the determinant of $\Sigma$. 
Proposition~\ref{thm:inf-loss} gives characteristics of the inference loss function.

\begin{proposition}\label{thm:inf-loss}
The Bayes rule with respect to the inference loss in (\ref{eq:inference-loss}) is $\hat d = E_{\Pi_n}(\vartheta\mid\bm{\mathsf{Z_n}})$ and $\hat\Sigma = Var_{\Pi_n}(\vartheta\mid\bm{\mathsf{Z_n}})$, which gives the minimized posterior risk $$E_{\Pi_n}\left\{L(\vartheta, \hat d, \hat\Sigma)\mid\bm{\mathsf{Z_n}}\right\}=\log \lvert Var_{\Pi_n}(\vartheta\mid\bm{\mathsf{Z_n}}) \rvert+p$$.
\end{proposition}

By Proposition~\ref{thm:inf-loss}, the Bayes rules for $d$ and $\Sigma$ with respect to $L_I$ are 
the posterior mean and variance. Bernstein-von Mises' Theorem implies that, under correct model specification, the posterior distribution is asymptotically equivalent to a Normal distribution, centered at the MLE and with variance equal to the (frequentist) sampling variance of the posterior mean, which in turn is the same sampling variance as the MLE in large samples~\cite[\S 10.2]{van2000asymptotic}, i.e. $\hat\Sigma_n = Var_{\Pi_n}(\vartheta\mid\bm{\mathsf{Z_n}})\approx Var_{P_{\theta_0}}(\hat d)$.  As a result, using the inference loss function produces both a Bayesian estimate ($\hat d$) of the parameter and  a measure $(\hat\Sigma)$ of variability in that estimate over repeated experiments -- i.e. the standard elements of a frequentist analysis. 

We note that the inference loss function is not new: it was previously discussed by \citet{dawid1999coherent}, who noted that it corresponds to using the log-determinant of the covariance matrix as a criterion to compare study designs. \citet{holland2015minimum} also considered the inference loss as an example of a `proper' loss function. However, to our knowledge, the simple and general form of its Bayes rule is not known, nor are the extensions we provide below.

\section{Balanced Inference Loss functions for model-robust variance estimation}
\label{sec:bil}
We now extend the inference loss function and present a Bayesian analog of robust variance estimators, as the Bayes rules for a loss which penalizes both the lack of model fit \textit{and} estimation error. Our loss function is motivated by the `balanced loss functions' originally proposed by \citet{zellner1994bayesian}. Balanced loss functions have since been developed considerably (see for example \citet{dey1999estimation}, \citet{jafari2012bayesian} and \citet{chaturvedi2014bayesian}), but all share a common form in which the loss is a weighted average of a term indicating lack-of-fit (in some way) and one indicating estimation error. The general form of the balanced loss function we will consider is
\begin{align}
&L_{BI}(\theta, d, \Sigma, \Omega) = \log\lvert\Sigma\rvert +\label{eq:balanced-inference}\\&\qquad \underbrace{(\theta - d)^T\Omega\Sigma^{-1}(\theta - d)}_{\text{Estimation error}} +\nonumber\\&\qquad \underbrace{\frac{1}{n}\sum_{i=1}^n \dot l_i(\theta)^T\left\{\Omega I_n(\theta)\right\}^{-1}\dot l_i(\theta)}_{\text{Lack of fit}},\nonumber
\end{align}
where $\dot l_i(\theta)=\frac{\partial}{\partial\theta}\log\left\{p_{\theta}(\bm Z_i)\right\}$ is the score function based on a single observation,   $I_n(\theta)=-\frac{1}{n}\sum_{i=1}^n\frac{\partial^2}{\partial\theta^2}\log\left\{p_\theta(\bm Z_i)\right\}$ is the (scaled) empirical Fisher information for $\theta$, and we make decisions  $\Sigma$ and $\Omega$ which are both $p\times p$ positive-definite weighting matrices, as well as decision $d\in\mathbb{R}^p$, an estimate of $\theta$. We call $L_{BI}$ in Equation~(\ref{eq:balanced-inference}) the \textit{balanced inference loss} function.

The first two terms in $L_{BI}$ are essentially the inference loss from Equation~(\ref{eq:inference-loss}), described earlier. However the balanced inference loss describes an additional decision, matrix $\Omega$, that sets rates of trade-off between elements of the estimation component of the inference loss versus the data-dependent penalty $$\dfrac{1}{n}\sum_{i=1}^n \dot l_i(\theta)^T I_n(\theta)^{-1}\dot l_i(\theta),$$ a form of `signal to noise' ratio for deviations from the model, indicating lack of fit. Specifically, the score $\dot l_i(\theta)$, the gradient of the log-likelihood with respect to the parameters, will have large entries when the data and corresponding model disagree, i.e. when $\theta$ is far from the peak of the likelihood. The (inverse) Fisher information matrix describes the uncertainty -- `noise', informally -- in the deviations given by the score. Adding these signal-to-noise terms across all observations ensures, straightforwardly, that model fit at all observations is reflected, and by taking their average (i.e. multiplying the total by $1/n$) the lack of fit term will neither vanish nor dominate in the large sample limit. Furthermore, the lack of fit term is invariant to one-to-one transformations of $\theta$. 

We now consider the Bayes rules for the balanced inference function.

\begin{proposition}
\label{thm:bbil}
The Bayes rule with respect to the inference loss in (\ref{eq:balanced-inference}) is  
\begin{align*}&\hat d = E_{\Pi_n}(\vartheta\mid\bm{\mathsf{Z_n}}), \\
&\hat\Omega=E_{\Pi_n}\left\{\,\frac{1}{n}\sum_{i=1}^n \dot l_i(\vartheta)\dot l_i(\vartheta)^TI_n(\vartheta)^{-1}\mid\bm{\mathsf{Z_n}}\,\right\}, \\
&\hat\Sigma = Var_{\Pi_n}(\vartheta\mid\bm{\mathsf{Z_n}})\hat\Omega.\end{align*}
\end{proposition}

From Proposition~\ref{thm:bbil} we see that the posterior mean remains the Bayes rule for estimate $d$, but the decision for $\Sigma$ is now a scaled version of the posterior variance, not the simple posterior variance from Proposition~\ref{thm:inf-loss}, where the rescaling term $\Omega$ is -- essentially -- the posterior expectation of the `lack of fit' terms in Equation~(\ref{eq:balanced-inference}), omitting $\Omega$.

In Theorem~\ref{thm:brse} we show that $\hat\Sigma$ 
estimates the large-sample variance of $\hat d$ in a model-robust manner. 

\begin{theorem}\label{thm:brse}

 Let $\theta^* = \argmin_\theta E_0\left\{p_0(\bm Z_i)/p_\theta(\bm Z_i)\right\}$ be the minimal Kullback-Leibler point. Under the regularity conditions in Section B of the Supplementary Material, we have
 $$n\hat\Sigma\overset{P_0}{\rightarrow}I(\theta^*)^{-1}E_0\{\dot l_i(\theta^*)\dot l_i(\theta^*)^T\}I(\theta^*)^{-1},$$
 where the right-hand side is also the asymptotic variance of $\hat d$ derived in~\cite{kleijn2012bernstein}:
 $$\sqrt n(\hat d - \theta^*)\overset{d}{\rightarrow}N(0, I(\theta^*)^{-1}E_0\{\dot l_i(\theta^*)\dot l_i(\theta^*)^T\}I(\theta^*)^{-1}).$$

\end{theorem}

Theorem~\ref{thm:brse} indicates that $\hat\Sigma$ is asymptotically equivalent to the asymptotic frequentist variance of $\hat d$ (i.e. over repeated experiments), under only mild regularity conditions that do not require fully-correct model specification, and is thus a Bayesian analog of the frequentist `robust' variance estimator. 

As well as providing a Bayesian approach to these empirically-useful methods, our work also allows us to construct Bayesian analogs of Wald-type confidence intervals. We define the two-sided Bayesian robust confidence interval at level $(1-\alpha)$ for the $j$th entry of $\hat d$ as $$\hat d_{j}\pm z_{1-\alpha/2}\hat\Sigma_{jj}^{1/2},$$
where $z_{1-\alpha/2}$ is the $(1-\alpha/2)$-quantile of the standard normal distribution and $\hat\Sigma_{jj}$ is the $j$th diagonal element of $\hat\Sigma$. We stress this is not in general a credible interval; its coverage properties are investigated in Section 4. In practice, it is valuable to report the Bayesian robust standard error ($\hat\Sigma_{jj}^{1/2}$) and Bayesian robust confidence intervals, which inform the frequentist uncertainty of estimates over repeated experiments. Comparing them with the common Bayesian summaries such as posterior standard deviation and credible intervals enables evaluation of the impact of model violations on statistical inference of the parameters.

Besides providing model-robust inference, we note that a slightly simpler version of the balanced inference loss function -- with a univariate correction term $\omega$ instead of a matrix-valued $\Omega$ -- provides a Bayesian analog of quasi-likelihood models~\citep{mccullagh1983quasi}. Details are provided in Section C of the Supplementary Material.

\section{Examples and simulation studies}
\label{sec:sim}
In this section we show, through analytic examples and  simulation studies, that in large samples the Bayesian robust variance estimate quantifies the variability, over repeated experiments, of the realized Bayes point estimates, and that it does so even in the presence of model mis-specification. Derivation of the results are deferred to Section D of the Supplementary Material.

\subsection{Estimating a Normal mean, but mis-specifying the variance}

As a deliberately straightforward first example, suppose a random sample $Y_1, \dots, Y_n$  is drawn from $N(\theta, \sigma_0^2)$ where mean $\theta$ is unknown mean and  variance $\sigma_0^2$ is assumed known. We are interested in estimating $\theta$ but wrongly assume the observations have variance $1$. In other words, we use the model $Y_i\mid\vartheta=\theta_0\sim N(\theta_0, 1)$, in which the variance is mis-specified. The analysis also uses prior  $\vartheta\sim N(\mu, \eta^2)$, for known $\mu, \eta$. Since $\theta$ is one-dimensional, we write $\Sigma=\sigma^2$ and $\Omega = \omega$.

The balanced loss function $L_{BI}$ in this setting is
\begin{align*}
&L_{BI}(\theta, d, \sigma^2, \omega) =\\ &\qquad \log\lvert\sigma^2\rvert + \omega(\theta - d)^2/\sigma^2 + \frac{1}{n\omega}\sum_{i=1}^n (Y_i - \theta)^2.
\end{align*}
Let $\bar Y = \sum_{i=1}^n 
  Y_i/n$ be the sample average. The Bayes rule sets 
\begin{align*}
&\hat d = \dfrac{\mu + n\eta^2 \bar Y}{n\eta^2 + 1},\\&\hat \omega = \dfrac{1}{n}\sum_{i=1}^n (Y_i - \bar Y)^2 + \dfrac{\eta^2}{n\eta^2 + 1} + \left(\dfrac{\mu - \bar Y}{n\eta^2 + 1}\right)^2,\\ &\hat\sigma^2 = \frac{\eta^2}{n(n\eta^2+1)}\sum_{i=1}^n (Y_i - \bar Y)^2 +  \left(\frac{\eta^2}{n\eta^2+1}\right)^2 + \\&\qquad \dfrac{\eta^2(\mu - \bar Y)^2}{(n\eta^2 + 1)^3}.
\end{align*}
Letting $n$ go to infinity, the point estimate is equivalent to the sample average $\bar Y$ and the variance estimate is equivalent to $\sigma_0^2/n$, i.e. the asymptotic inference that would be obtained under a correct model specification. If we had instead relied on the assumed model, then using inference loss~(\ref{eq:inference-loss}) the variance estimate is equivalent to $1/n$, i.e wrong by a scale factor of $\sigma_0^2$. The underlying approach remains fully parametric, and retains the appealing normative element of decision theory, where a clear statement about the goal of the analysis (here, balancing estimation of $\theta$ with fidelity to the data) makes the our choice of what to report (i.e. $\hat d$) both automatic and optimal.



\subsection{Estimating the variance of Generalized Linear Model regression coefficient estimates}
  \label{sec:glm}
The balanced inference loss function takes an appealingly straightforward form when the data are independent observations $Y_1, \dots, Y_n$, from a generalized linear model (GLM), i.e. an exponential family where the distribution has the density 
$$p_{\theta, \alpha}(Y_i)= \exp\left(\dfrac{Y_i\theta_i - b(\theta_i)}{\alpha}+c(Y_i, \alpha)\right)$$
for functions $b(\cdot)$, $c(\cdot, \cdot)$ and scalars $\theta$ and $\alpha$~\citep{mccullagh1989generalized}. We assume the presence of $p$-dimensional explanatory variables $\bm X_i$ augmenting each observation $y_i$, with a link function $g(\cdot)$ connecting the mean function $\mu_i = E_{P_{\theta, \alpha}}(Y_i)=\partial b(\theta_i)/\partial\theta_i$ and the linear predictors $\bm X_i^T\beta$ via $g(\mu_i)=\bm X_i^T\beta$, where $\beta$ is a $p$-vector of regression coefficients. We write $\bm Y = (Y_1, \dots, Y_n)^T$ and $\bm{\mathsf{X}} = (\bm X_1, \dots, \bm X_n)^T$. For ease of development, we focus on inference for $\beta$ and assume the dispersion parameter $\alpha$ is known. In the two examples below, this assumption does not alter the results, either because $\alpha$ is a fixed at a known constant or because the two parameters are orthogonal~\citep{cox1987parameter}. For simplicity, here we will not distinguish between parameters used in frequentist or Bayesian analyses.

With a GLM, the score function with respect to $\beta$ for a single observation is
$$\dot l_i(\beta) = \dfrac{\partial \mu_i}{\partial\beta}\dfrac{Y_i - \mu_i}{\alpha V_i}$$ and the empirical Fisher information is 
\begin{align*}
    &I_n(\beta) 
    = -\dfrac{1}{n}\sum_{i=1}^n \bigg\{\dfrac{\partial^2\mu_i}{\partial\beta\partial\beta^T}\left(\dfrac{Y_i - \mu_i}{\alpha V_i}\right) -\\&\quad \dfrac{1}{\alpha V_i}\dfrac{\partial\mu_i}{\partial\beta}\left(\dfrac{\partial\mu_i}{\partial\beta}\right)^T - \dfrac{\partial\mu_i}{\partial\beta}\left(\dfrac{Y_i - \mu_i}{\alpha V_i^2}\right)\left(\dfrac{\partial V_i}{\partial\beta}\right)^T\bigg\},
\end{align*}
where $V_i = \partial\mu_i/\partial\theta_i$. When the canonical link function is use, i.e. $\theta_i=X_i^T\beta$, $\dot l_i(\beta)$ and $I_n(\beta)$ can be considerably simplified as
\begin{align*}
    \dot l_i(\beta) = X_i(Y_i - \mu_i)/\alpha, \quad I_n(\beta) = \dfrac{1}{n}\sum_{i=1}^n X_iX_i^TV_i/\alpha.
\end{align*}
Therefore, with the canonical link, the balanced inference loss function is

\begin{align}\label{eq:bil-glm}
&L_{BI}(\beta,d, \Sigma, \Omega) =\log\mid\Sigma\mid + \\&\qquad(\beta - d)^T\Omega\Sigma^{-1}(\beta - d) +  \nonumber\\&\qquad\sum_{i=1}^n \frac{(Y_i - \mu_i)^2}{\alpha V_i}\cdot \bm{X}_i^T\left\{\Omega \bm{\mathsf X}^T\bm{\mathsf X}\right\}^{-1}\bm X_i,\nonumber\end{align}
in which we observe that the outcome data enter the balancing `lack of fit' term via the terms $(Y_i-\mu_i)^2/\alpha V_i$, which can be thought of as a Bayesian analog of (squared) Pearson residuals. The other components in the balancing term determine how much weight these `squared' terms receive. 
The balanced inference loss's Bayes rule sets $d$ to be the usual posterior mean, and 
\begin{align*}
&\hat\Sigma_n = Var_{\Pi_n}(\beta\mid\bm{\mathsf{Z_n}})E_{\Pi_n}\big\{\bm{\mathsf{X}}^T\bm{\mathsf{S}}\bm{\mathsf{X}}\left(\bm{\mathsf{X}}^T\bm{\mathsf{V}}\bm{\mathsf{X}}\right)^{-1} / \alpha\mid\bm{\mathsf{Z_n}}\big\},\end{align*}
where $\bm{\mathsf V}$ is the $n\times n$ diagonal matrix with the $i$-th diagonal element $V_i$ and $\bm{\mathsf{S}}$  with elements $(Y_i - \mu_i)^2$. $\hat\Sigma_n$ can be thought of as the model-based posterior variance `corrected' by a term that assesses fidelity of the data to the model. As with the general case, asymptotic equivalence with the classical `robust' approach holds.

\subsubsection*{Example: Linear regression} 
~~For observations $Y_1, \dots, Y_n$ and the explanatory variables $\bm X_1, \dots, \bm X_n$, a Bayesian linear regression model assumes $$Y_i\mid\bm X_i, \beta, \sigma^2\sim N(\bm X_i^T\beta, \sigma^2)$$where $\sigma^2$ is unknown. The corresponding balanced inference loss is 
\begin{align*}
&L_{BI}(\beta, d, \Sigma, \Omega) = \log\mid\Sigma\mid + (\beta - d)^T\Omega\Sigma^{-1}(\beta - d) +\\
&\qquad \sum_{i=1}^n \frac{(Y_i - \bm X_i^T\beta)^2}{\sigma^2}\cdot \bm X_i^T\left\{\Omega\sum_{j=1}^n\bm X_j \bm X_j^T\right\}^{-1}\bm X_i,\end{align*}in which the novel penalty term is a sum of squared Pearson errors, weighted by a modified form of leverage~(\cite{mccullagh1989generalized} Section 12.7.1, page 405) which includes a term $\Omega$ inside the diagonal elements of the familiar `hat' matrix $\bm{\mathsf X} (\bm{\mathsf X}^T\bm{\mathsf X})^{-1} \bm{\mathsf X}^T$. The Bayes rule for $\Sigma$ is the Bayesian robust variance 
\begin{align*}
&\hat\Sigma_n = Var_{\Pi_n}(\beta\mid\bm{\mathsf{Z_n}})E_{\Pi_n}\left\{\bm{\mathsf{X}}^T\bm{\mathsf{S}}\bm{\mathsf{X}}(\bm{\mathsf{X}}^T\bm{\mathsf{X}})^{-1}/\sigma^2\,\mid\,\bm{\mathsf{Z_n}}\right\}.\end{align*}

We perform a simulation study to investigate the performance of Bayesian robust variance for linear regression. We generate univariate explanatory variables $\bm X_1, \dots, \bm X_n$ as $\bm X_i = (1, U_i)^T$ where $U_i\sim U(0, 3)$ independently, and generate outcome variables as $Y_i\mid\bm X_i \sim N(U_i + a U_i^2, 1)$ for $i=1,\dots, n$, for various constants $a$. When $a\neq 0$, the mean model, $E(Y\mid\bm X=\bm x)$ is quadratic in $\bm x$, and hence the model is misspecified. 

We suppose interest lies in linear trend parameter $\beta_2$ and $\hat d$ is the posterior mean, i.e. the Bayes rule for its estimation. We use weakly informative priors $\beta_j\sim N(0, 10^3)$ for $j=1,2$ and $\sigma^{-2}\sim \text{Gamma}(0.1, 0.1)$. For 1000 simulations under each setting, we report the average  posterior mean ($Ave(\hat d)$), the standard error of $\hat d$ ($SE(\hat d)$), the average posterior standard deviation of $\vartheta$ ($Ave(SD_{\Pi_n}(\vartheta\mid\bm{\mathsf{Z_n}}))$), and the average Bayesian robust standard error estimate ($Ave(\widehat{BRSE}(\hat d))$).

Table~\ref{tab:brse-lm} compares the posterior standard deviation and the Bayesian robust standard error estimate, as estimates of the standard error of the Bayes rule. In all the scenarios except when $a=0$ (noted in grey, indicating that the linear regression model is correctly specified), while there is a recognizable difference between the true standard error and the posterior standard deviation, the Bayesian robust standard error estimate is notably closer to the true standard error. 

\begin{table}[!htbp]
\caption{\label{tab:brse-lm} Comparison between posterior standard deviation and the Bayesian robust standard error estimate to the true standard error of the Bayes point estimate $\hat d$ in linear regression. Ave($\hat d$): Monte Carlo average of $\hat d$; $SE(\hat d)$: Monte Carlo standard deviation of $\hat d$; Ave(Post.SD): average posterior standard deviation; Ave($\widehat{BRSE}$): Monte Carlo average of Bayesian robust standard error. Gray rows indicate where the model is correctly specified.}
\centering
\resizebox{0.98\linewidth}{!}{
    \begin{tabular}{rr|rrrrrrr}
  \hline
n & $a$ & Ave($\hat d$) & $SE(\hat d)$ &  Ave(Post.SD) & $Ave(\widehat{BRSE})$ \\ 
  \hline
50 & -2 & -4.996 & 0.332 &  0.282 &  0.331 \\ 
  50 & -1 & -2.000 & 0.216 &  0.203 & 0.220 \\ 
\rowcolor{Gray}  50 & 0 & 1.009 & 0.171 &  0.167 & 0.167 \\ 
  50 & 1 & 3.999 & 0.227 &  0.203 &  0.220 \\ 
  50 & 2  & 6.997 & 0.349 &  0.282 &  0.334 \\ 
  100 & -2  & -4.980 & 0.233 &  0.198 &  0.233 \\ 
  100 & -1 & -2.001 & 0.150 &  0.141 &  0.153 \\ 
\rowcolor{Gray}   100 & 0  & 1.001 & 0.115 &  0.117 &  0.117 \\ 
  100 & 1 & 4.008 & 0.152 & 0.141 & 0.153 \\ 
  100 & 2  & 6.994 & 0.223 &  0.197 &  0.231 \\
   \hline
\end{tabular}}
\end{table}
\subsubsection*{Example: Poisson regression} 

~~Poisson regression is a default method for studying the association between count outcomes and the explanatory variables~\citep{agresti2012categorical}. With count observations $Y_1, \dots, Y_n$ and corresponding  explanatory variables $\bm X_1, \dots, \bm X_n$  (which may contain an intercept) a Poisson regression model assumes $Y_i\mid\bm X_i \sim \mathrm{Poisson}\left(\exp(\bm X_i^T\beta)\right)$. The balanced inference loss is
\begin{small}\begin{align*}&L_{BI}(\beta, d, \Sigma, \Omega) = \log\mid\Sigma\mid + (\beta - d)^T\Omega\Sigma^{-1}(\beta - d) +\\&\sum_{i=1}^n \{Y_i - \exp(\bm X_i^T\beta)\}^2\times \\&\quad \bm X_i^T\left\{\Omega\sum_{j=1}^n\bm X_j\bm X_j^T\exp(\bm X_j^T\beta)\right\}^{-1}\bm X_i,\end{align*} 
\end{small}
which is again interpretable as a weighted sum of squared Pearson errors, and the weights are again modified versions of leverage. The Bayes rule for $\Sigma$ is another Bayesian robust variance estimate, 
\begin{align*}
&\hat\Sigma_n = Var_{\Pi_n}(\beta\mid\bm{\mathsf{Z_n}})E_{\Pi_n}\bigg\{\bm{\mathsf{X}}^T\bm{\mathsf{S}}\bm{\mathsf{X}}\left(\bm{\mathsf{X}}^T\bm{\mathsf{V}}\bm{\mathsf{X}}
\right)^{-1}\mid \bm{\mathsf{Z_n}}\bigg\},
\end{align*}
where the diagonal elements of $\bm{\mathsf V}$ are $\exp(\bm X_i^T\beta)$.

To evaluate the method, as before we generate explanatory variables $\bm X_i = (1, U_i)^T$, where $U_i\sim U(-3, 3)$, and generate the outcomes as $Y_i\mid\bm X_i \sim \mathrm{Poisson}(U_i + a U_i^2)$, for various constants $a$ and $i=1,\dots, n$. The model is mis-specified when $a\neq 0$. We assume the parameter of interest is the log risk ratio $\theta = \beta_2$ and $\hat d$ is the posterior mean, the corresponding Bayes rule. We use weakly informative priors $\beta_j\sim N(0, 10^3)$ for $j=1,2$. For each setting we perform 1,000 replicates, and report the same measures as for linear regression.

Table \ref{tab:brse-poisson} shows the results. Again, the Bayesian robust standard error estimate is a notably better estimate of the true standard error of $\hat d$ than the posterior standard deviation of $\vartheta$; it performs better when the model is misspecified and no worse when correctly specified.

\begin{table}[!htbp]
\centering
\caption{\label{tab:brse-poisson} For a Poisson model, comparison of the posterior standard deviation and the Bayesian robust standard error estimate to the true standard error of the posterior mean $\hat d$. Ave($\hat d$): Monte Carlo average of $\hat d$; $SE(\hat d)$: Monte Carlo standard deviation of $\hat d$; Ave(Post.SD): average posterior standard deviation; Ave($\widehat{BRSE}$): Monte Carlo average of Bayesian robust standard error. Gray rows indicate where the model is correctly specified.}
\resizebox{0.98\linewidth}{!}{
\begin{tabular}{rr|rrrrrrr}
  \hline
n & $a$ & Ave($\hat d$) & $SE(\hat d)$ &  Ave(Post.DF) & Ave($\widehat{BRSE}$) \\ 
  \hline
50 & -0.50 &  0.348 & 0.101 &  0.115 & 0.100 \\ 
  50 & -0.25 &  0.576 & 0.091 &  0.098  & 0.090 \\ 
\rowcolor{Gray}  50 & 0.00 & 1.014 & 0.086 &  0.085  & 0.084 \\ 
  50 & 0.25 &  1.788 & 0.134 &  0.070  & 0.118 \\ 
  50 & 0.50 &  2.870 & 0.296 &  0.049 & 0.181 \\ 
  100 & -0.50 &  0.344 & 0.066 &  0.080  & 0.068 \\ 
  100 & -0.25 &  0.567 & 0.060 &  0.068  & 0.062 \\ 
\rowcolor{Gray}  100 & 0.00 & 1.005 & 0.059 & 0.058  & 0.058 \\ 
  100 & 0.25 &  1.808 & 0.089 &  0.049 & 0.084 \\ 
  100 & 0.50 &  2.954 & 0.163 & 0.034  & 0.136 \\ 
  \hline
\end{tabular}}
\end{table}

\subsection{Estimating the variance of log hazards ratio with an exponential proportional hazards model}\label{sec:eph}

We demonstrate how our framework can be extended beyond GLMs with an association analysis for time-to-event outcomes and  explanatory variables. 

Suppose the observed data are $(\widetilde T_i, \bm X_i, \Delta_i)$ for $i=1,\dots, n)$, where $\widetilde T_i = \min(T_i, C_i)$ is the observed survival time, $T_i$ is the true survival time, $\Delta_i = I(T_i \leq C_i)$ is the event indicator,  $C_i$ is the censoring time and $\bm X_i$ is a $p$-vector of covariates. 

Suppose we intend to model the association between the survival time and the covariates using the exponential proportional hazards model $$T_i\mid\bm X_i,\beta\sim \mathrm{Exponential}(\exp(\bm X_i^T\beta)).$$  

The balanced inference loss is
\begin{align*}&L_{BI}(\beta, d, \Sigma, \Omega) = \log\mid\Sigma\mid + (\beta - d)^T\Omega\Sigma^{-1}(\beta - d) + \\&\qquad \sum_{i=1}^n (\Delta_i - \widetilde T_i\exp(\bm X_i^T\beta))^2 \times \\&\qquad \bm X_i^T\left[\Omega\left\{\sum_{j=1}^n \widetilde T_j\exp(\bm X_j^T\beta)\bm X_j\bm X_j^T\right\}\right]^{-1}\bm X_i\end{align*}

In this case, the interpretation of the balancing term is less clear than before. We note, however, for those subjects who experienced  events (i.e. $\Delta_i = 1$), their contribution to $L_{BI}$ can be written as
\begin{align*}&\dfrac{\{\widetilde T_i - \Delta_i\exp(-\bm X_i^T\beta)\}^2}{\exp(-2\bm X_i^T\beta)}\times \\&\qquad \bm X_i^T\left[\Omega\left\{\sum_{j=1}^n \widetilde T_j\exp(\bm X_j^T\beta)\bm X_j\bm X_j^T\right\}\right]^{-1}\bm X_i,\end{align*}
which, again, is their Pearson residual weighted by some version of their leverage.

The Bayes rule for $\Sigma$ is a Bayesian robust covariance matrix 
\begin{align*}&\hat \Sigma_n =Var_{\Pi_n}(\beta\mid\bm{\mathsf{Z_n}})E_{\Pi_n}\bigg\{\bm {\mathsf X}^T\bm{\mathsf S}\bm{\mathsf{X}}\left(\bm{\mathsf{X}}^T\bm{\mathsf V}\bm{\mathsf{X}}\right)^{-1}\mid\bm{\mathsf{Z_n}}\bigg\}\end{align*}
where $\bm{\mathsf S}$ is an $n\times n$ diagonal matrix with elements $\{\Delta_i - \widetilde T_i\exp(\bm X_i^T\beta)\}^2$ and $\bm{\mathsf V}$  with elements $\widetilde T_i\exp(\bm X_i^T\beta)$.

In a simulation study to investigate the robustness of different variance estimates in exponential proportional hazards models, we generated the covariates as $\bm X_i = (1, U_i)^T$ where $U_i\sim U(0, 3)$, for $i=1,\dots, n$, and the survival times by a Weibull proportional hazards model $$T_i\mid\bm X_i, \beta \sim \mathrm{Weibull}(\kappa, \exp(\bm X_i^T\beta)),$$
where the pdf of a $\mathrm{Weibull}(\kappa, \lambda)$ random variable is given by
$$f(t) = \lambda\kappa t^{\kappa-1}\exp(-\lambda t^\kappa),\qquad t>0.$$ We further performed administrative censoring at $C_i=10$, $i=1,\dots, n$. 
Our use of Exponential proportional model deviates from the true data generating distribution if $\kappa\neq 1$.

We suppose the parameter of interest is the log hazards ratio $\theta = \beta_2$ in the model and use the prior distribution $\beta_,\sim N(0, 10^ 3)$, $j=1,2$. Again, we conduct 1000 simulations under each setting considered, and report the same measures as before.

Results are given in Table~\ref{tab:brse-expph}. Although slightly underestimating the standard error of the Bayes estimator $\hat d$ when the model is correctly specified, in general the Bayesian robust standard error estimate is a better estimate of the true standard error of $\hat d$  than the posterior standard deviation of $\vartheta$.

\begin{table}[!htbp]
\centering
\caption{\label{tab:brse-expph} For an Exponential Proportional Hazards Model, comparison of the posterior standard deviation and the Bayesian robust standard error to the true standard error of the posterior mean $\hat d$. \#E: Average number of events; Ave($\hat d$): Monte Carlo average of $\hat d$; $SE(\hat d)$: Monte Carlo standard deviation of $\hat d$; Ave(Post.SD): average posterior standard deviation; Ave($\widehat{BRSE}$): Monte Carlo average of Bayesian robust standard error estimate. Gray rows indicate where the model is correctly specified. }
\resizebox{0.98\linewidth}{!}{
\begin{tabular}{rrrr|rrrrrrr}
  \hline
n & $\kappa$ & $\beta$ & \#E & Ave($\hat d$) & $SE(\hat d)$ &  Ave(Post.SD) & $Ave(\widehat{BRSE})$ \\ 
  \hline

  50 & 0.8 & 0.00 & 49.9 & -0.004 & 0.187 & 0.151 & 0.164 \\ 
  50 & 0.8 & -0.25 & 49.8 & -0.317 & 0.179 & 0.150 & 0.165 \\ 
  50 & 0.8 & -0.50 & 49.4 & -0.608 & 0.185 & 0.152 & 0.165 \\ 
\rowcolor{Gray}  50 & 1.0 & 0.00 & 50.0 & 0.000 & 0.146 & 0.148 & 0.132 \\ 
\rowcolor{Gray}  50 & 1.0 & -0.25 & 49.8 &  -0.250 & 0.151 & 0.149 & 0.134 \\ 
 \rowcolor{Gray} 50 & 1.0 & -0.50 & 49.9& -0.503 & 0.146 & 0.149 & 0.134 \\ 
  50 & 1.5 & 0.00 & 50.0&  0.006 & 0.101 & 0.148 & 0.094 \\ 
  50 & 1.5 & -0.25 & 50.0&  -0.165 & 0.101 & 0.148 & 0.093 \\ 
  50 & 1.5 & -0.50 & 50.0& -0.337 & 0.102 & 0.146 & 0.093 \\ 
  100 & 0.8 & 0.00 & 99.8& 0.003 & 0.130 & 0.103 & 0.119 \\ 
  100 & 0.8 & -0.25 & 99.6 & -0.311 & 0.131 & 0.103 & 0.119 \\ 
  100 & 0.8 & -0.50 & 98.8 & -0.614 & 0.128 & 0.105 & 0.119 \\ 
 \rowcolor{Gray} 100 & 1.0 & 0.00 & 100.0 & -0.006 & 0.102 & 0.102 & 0.095 \\ 
\rowcolor{Gray}  100 & 1.0 & -0.25 & 100.0 & -0.248 & 0.102 & 0.103 & 0.096 \\ 
\rowcolor{Gray}  100 & 1.0 & -0.50 & 99.8 & -0.497 & 0.104 & 0.102 & 0.097 \\ 
  100 & 1.5 & 0.00 & 100.0 & -0.001 & 0.071 & 0.102 & 0.066 \\ 
  100 & 1.5 & -0.25 & 100.0 & -0.170 & 0.070 & 0.102 & 0.066 \\ 
  100 & 1.5 & -0.50 & 100.0&  -0.332 & 0.072 & 0.102 & 0.067 \\ 
  \hline
\end{tabular}}
\end{table}

\subsection{Bayesian robust confidence intervals}

In previous sections, we have shown that (for several commonly used regression models) our Bayesian robust standard error estimate is approximately unbiased for the true standard errors of model parameter estimators over the minimal Kullback-Leibler point. We further investigate whether they can be used to construct well-calibrated confidence intervals. Using the posterior mean as estimate $\hat d$ and Bayesian robust variance estimate $\hat\Sigma$, we construct the Bayesian analog of Wald-type `robust' confidence intervals, as described by \citet{royall1986model}. The proposed Bayesian robust confidence interval at confidence level $1-\alpha$ for $\theta_j$, $j=1,\dots, p$, is
$$\hat d_{j}\pm z_{1-\alpha/2}\hat \Sigma_{jj}^{\frac{1}{2}}.$$

Figure \ref{fig:int_cov} shows the coverage probabilities of these intervals, together with standard model-based 95\% credible intervals and standard robust confidence intervals, using the examples in Section \ref{sec:glm} and \ref{sec:eph}. We assume the same models and data distributions as in the previous sections. In all the model-misspecification scenarios, the posterior credible interval does not -- as expected -- have the nominal level coverage even in large samples. Both standard frequentist robust confidence intervals and Bayesian robust confidence intervals are well calibrated.

\begin{figure*}
    \centering
    \begin{tabular}{c}
    \includegraphics[width = 0.75\linewidth]{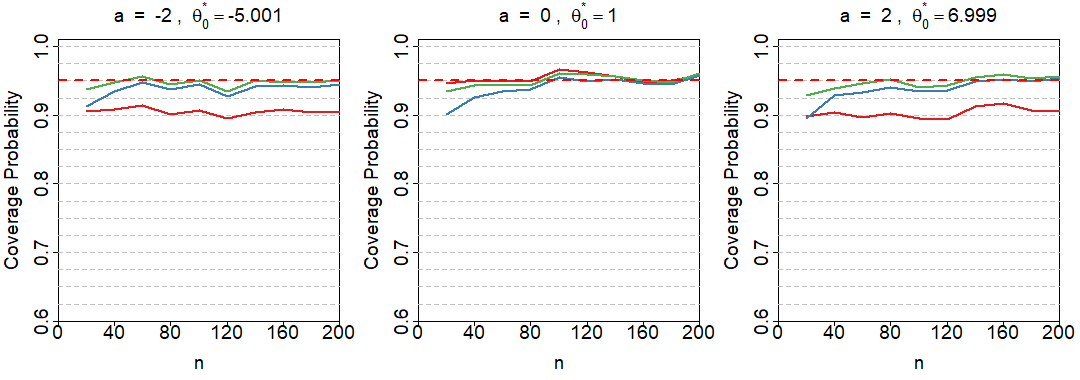}\\
    (a)\\
    \includegraphics[width = 0.75\linewidth]{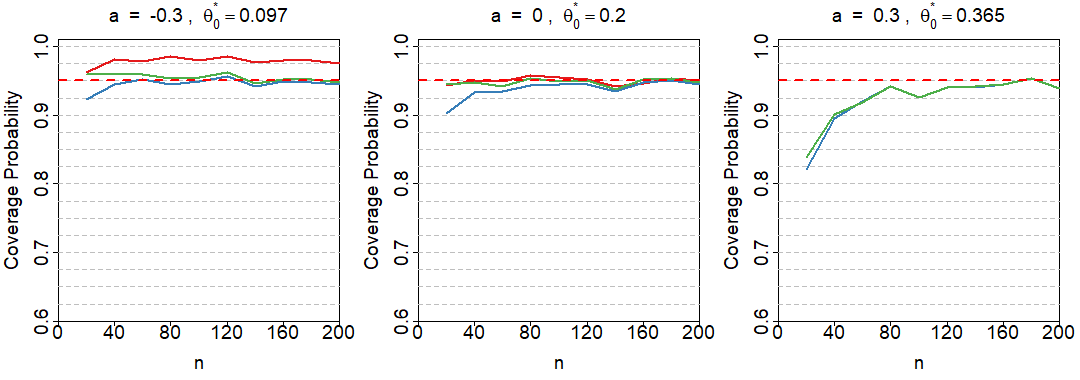}\\
    (b)\\
    \includegraphics[width = 0.75\linewidth]{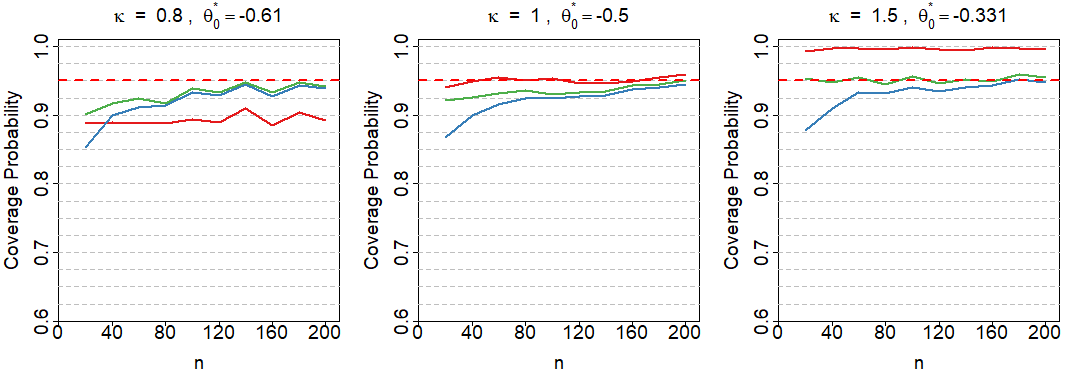}\\
    (c)
    \end{tabular}
    \caption{Coverage probabilities of 95\% credible interval (red), frequentist (blue) and Bayesian (green) robust confidence intervals for (a) linear regression, (b) Poisson regression, and (c) exponential proportional hazards model. The middle column shows the results with  correctly-specified models. $\theta_0^*$ denotes the minimal Kullback-Leibler point in each scenario.}
    \label{fig:int_cov}
\end{figure*}

\section{Application: Association of Systolic Blood Pressure and Age in 2017-2018 NHANES Data}
\label{sec:app}
We demonstrate the Bayesian robust standard error estimator using data from 2017-2018 National Health and Nutrition Examination Survey (NHANES) data. We are interested in the association between subjects' systolic blood pressure systolic and their age among the subpopulation with age between 18 and 60 years old. To make the difference between the variance estimates more distinguishable, we randomly select a subset of 200 subjects for our analysis.

We used a simple linear regression model as a working model, with the average of two systolic blood pressure readings as the outcome and age and gender as covariates (Figure \ref{fig:nhanes-methods} left panel). We computed both the model-based standard error estimates (assuming homoscedasticity) and robust standard error estimates (resistent to heteroscedasticity and/or model misspecification) of the regression coefficients. For the corresponding Bayesian model, we postulate a normal linear regression model with the same outcome and covariates. We use weakly informative prior distributions for the unknown regression coefficients and the residual variance (Figure \ref{fig:nhanes-methods} right panel). We report the posterior mean as point estimates, and compute the posterior standard deviations and Bayesian standard errors for the regression coefficients. For the Bayesian analysis, we used Gibbs sampling in JAGS~\citep{plummer2003jags} with 3 chains and 30,000 iterations for each chain, where the first 18,000 iterations in each chain are burn-in and discarded from the analysis. 

\begin{figure*}[!htbp]
    \centering
    \begin{tabular}{c|c}
    Frequentist & Bayesian \\
    \hline
    \begin{minipage}{0.32\linewidth}
    $$E(\mathrm{SBP}) = \beta_0 + \beta_1 \mathrm{MALE} + \beta_2\mathrm{AGE}$$\end{minipage} &
    \begin{minipage}{0.32\linewidth}
  $$\mathrm{SBP}\sim N(\beta_0 + \beta_1 \mathrm{MALE} + \beta_2\mathrm{AGE}, \sigma^2)$$
    $$\beta_j\sim N(0, 1000),\quad j=0,1,2$$
    $$\sigma^2 \sim InvGamma(0.01, 0.01)$$
    \end{minipage}
    \end{tabular}
    \caption{Frequentist and Bayesian regression models for analyzing the association between systolic blood pressure and age, adjusting for subjects' gender. SBP: systolic blood pressure; MALE: the indicator for male; AGE: subjects age in years.}
    \label{fig:nhanes-methods}
\end{figure*}

In Table~\ref{tab:nhanes-results}, we report the point estimates and the variance estimates of the regression coefficients for the frequentist and the Bayesian linear regression models. We focus on the comparison of the variance estimates. The model-based standard errors/Bayesian posterior standard deviations given higher estimates of the standard error for all three regression coefficients than the frequentist/Bayesian robust standard errors. For the intercept and the regression coefficient for age, the Bayesian posterior standard deviations are nearly equal to the corresponding frequentist model-based standard error estimates, whereas the Bayesian robust standard error estimate are approximately equal to the frequentist robust standard error estimate. The four versions of standard error estimates for the regression coefficient of gender are not as distinguishable.

\begin{table}[!htbp]
\centering
\caption{Point estimates and standard error estimates of regression coefficients in Frequentist and Bayesian linear regression models using the NHANES data. Post. SD: posterior standard deviation; BRSE: Baysian robust standard error estimate.}
\label{tab:nhanes-results}
\resizebox{0.98\linewidth}{!}{
\begin{tabular}{r|rrr|rrr}
  \hline
  &\multicolumn{3}{c|}{Frequentist} & \multicolumn{3}{c}{Bayesian}\\
 & Est. & SE & Robust SE & Est. & Post. SD & BRSE \\ 
  \hline
(Intercept) & 94.067 & 2.350 & 1.781 & 93.481 & 2.307 & 1.767 \\ 
  Male & 4.817 & 2.063 & 2.032 & 5.005 & 2.062 & 2.050 \\ 
  Age (yrs) & 0.570 & 0.046 & 0.042 & 0.580 & 0.046 & 0.042 \\
   \hline
\end{tabular}}
\end{table}

Code for the analysis and the dataset are available on GitHub (\url{https://github.com/KenLi93/BRSE}).

\section{Discussion}
\label{sec:discussion}
In this article, we have proposed a balanced inference loss function, whose Bayes rules provide a Bayesian analog of the robust variance estimator. We have proven how, in large samples and under only mild regularity conditions, the Bayesian robust variance estimate converges to the true variance, over repeated experiments, of the estimator of interest. In our examples of GLMs and other forms of models, the corresponding loss function is seen to be a straightforward balance between losses for standard model-based inference, and a term that assesses how well the corresponding model actually fits the data that have been observed; the penalty employs terms and components that are familiar from diagnostic use of Pearson residuals, and from study of leverage. Using these methods we also proposed the Wald confidence intervals using Bayesian robust standard error, and found attractively stable behavior in simulation. 

Our development of the robust variance estimates is fully parametric, pragmatically adjusting the standard model-based inference with a measure of how the corresponding model fits the data. Our proposed method preserves the same parameter of interest during the modeling and estimation and the interpretation of the point estimates stays the same as in regular parametric statistics. The proposed Bayesian variance estimates do not require nonparametric Bayes components in the model specification as in \citet{szpiro2010model} or additional second-order moments as in \citet{hoff2013bayesian}, thus reducing the computation time and increasing the algorithmic stability. Our framework is general and can be applied to obtain the Bayesian robust variance for any regular parametric models.

We should note that the data-dependent nature of the lack of fit term in our loss function places the proposed balanced inference loss function outside of classical decision theory~\citep{wald1949statistical,savage1951theory,parmigiani2009decision}.  Data-dependent losses have been developed for other uses~\citep{zellner1994bayesian,barron2019general,baik2021meta}  but do not automatically enjoy the coherence properties of Bayes rules from losses involving only parameters and decisions. Addressing how the degree and form of data-dependency in losses affects coherence and consistency properties of resulting Bayes rules remains an open and interesting research area.

We discussed regression models with random covariates in Section~\ref{sec:sim}, but our framework does not entertain having fixed covariates, i.e. considering random variability due to the outcomes alone and conditioning on the covariate values. For a full discussion of the impact of fixed versus random covariates, see \citet{buja2019modelsi} and \citet{buja2019modelsii}. We note, however, that model mis-specification when combined with random covariates, contributes an asymptotically non-negligble component of the variance of the estimates~\citet[\S 7.3]{buja2019modelsii}. Standard sandwich estimates account for this component, so using them when considering covariates are fixed tends to bias them upwards, i.e. makes them conservative. The degree of conservatism tends to be minor, but for robust frequentist methods tailored to fixed covariate settings see e.g. \citet{abadie2014inference}.

A general Bayesian analog of such methods remain undeveloped. (For linear regression with fixed covariates, \citet{szpiro2010model} extend their Bayesian sandwich by additionally assuming a categorical model for the fixed regressors -- but this approach will be ineffective under our approach since the likelihood of the regression coefficients do not depend on the marginal density function of the regressors.) In Section E of the Supplementary Material we show by simulation that the mild conservatism noted above does seem to carry over to Bayesian robust standard errors, in settings like that of Section~\ref{sec:glm} but with fixed regressors. This suggests that developing Bayesian analog of a fixed-covariate sandwich, while an interesting  challenge, is unlikely to drastically alter the results in any given analysis.

An area of potential development that seems more fruitful is generalizing the balanced inference loss~(\ref{eq:balanced-inference}) to other measures of discrepancy, both between parameter and estimate as well as in the model fit term. We have used forms of squared error loss in both roles, but absolute or $L_1$ losses can be expected to provide methods that are less influenced by the largest discrepancies over which we average. This robustness to extreme values (studied in depth by e.g. \citet{huber2004robust}) is distinct from the robustness to model assumptions we achieve, but both have appeal. Two major difficulties facing the use of $L_1$ losses include computation and the lack of an obvious way to generalize them from one-dimensional absolute losses to $L_1$ penalties over $p$ parameters, which will typically be correlated in the posterior.

Our work in this paper demonstrates that robust standard error estimates can be ubiquitous both in frequentist and in Bayesian statistics. Our approach enables the proper quantification of the variability of parameter estimates in Bayesian parametric models. The proposed balance inference loss function, through which the Bayesian robust standard error estimate is derived, also provides insights on the source of discrepancy between the model-based and model-agnostic variance estimates.

\bibliographystyle{imsart-nameyear} %
\bibliography{references}       %

\end{document}


\setcounter{table}{0}
\renewcommand{\thetable}{S\arabic{table}}%
\setcounter{figure}{0}
\renewcommand{\thefigure}{S\arabic{figure}}%
\setcounter{section}{0}
\renewcommand{\thesection}{\Alph{section}}
\setcounter{assumption}{0}
\renewcommand{\theassumption}{S\arabic{assumption}}
\setcounter{model}{0}
\renewcommand{\themodel}{S\arabic{model}}
\setcounter{theorem}{0}
\renewcommand{\thetheorem}{S\arabic{theorem}}
\setcounter{lemma}{0}
\renewcommand{\thelemma}{S\arabic{lemma}}
\setcounter{equation}{0}
\renewcommand{\theequation}{S\arabic{equation}}

\maketitle
\tableofcontents
\newpage
\section{Proof of Theorems}\label{append:proof}
We write $\mathbbm{1}$ as the indicator function: $\mathbbm 1(A)=1$ if $A$ is true and $0$ if otherwise.
\subsection{Proof of Proposition~2.1}\label{append:inf-loss}
The posterior risk is \begin{align*}
    &E_{\Pi_n}(L\mid\bm{\mathsf{Z_n}}) \\
    =&\log\mid\Sigma\mid + E_{\Pi_n}\left\{(\vartheta - d)^T\Sigma^{-1}(\vartheta - d)\mid\bm{\mathsf{Z_n}}\right\}\\
    =&\log\mid\Sigma\mid+tr\left[\Sigma^{-1}E_{\Pi_n}\left\{(\vartheta - d)(\vartheta - d)^T\mid\bm{\mathsf{Z_n}}\right\}\right].
\end{align*}
Writing $A = \Sigma^{-1}E_{\Pi_n}\left\{(\vartheta - d)(\vartheta - d)^T\mid\bm{\mathsf{Z_n}}\right\}$, the posterior risk is 
$$\log \lvert E_{\Pi_n}\left\{(\vartheta - d)(\vartheta - d)^T\mid\bm{\mathsf{Z_n}}\right\} \rvert - \log\lvert A \rvert + tr(A)$$
in which the last two terms are respectively the sum of the $-\log$ eigenvalues and the eigenvalues of $A$. This is minimized by setting $A$ so that all eigenvalues equal to 1, which occurs if and only if $A$ is the identity matrix. Hence $\Sigma = E_{\Pi_n}\left\{(\vartheta - d)(\vartheta - d)^T\mid\bm{\mathsf{Z_n}}\right\}$.

To minimize the posterior risk with respect to $d$, it remains to consider
\begin{align*}
&\log\lvert E_{\Pi_n}\left\{(\vartheta - d)(\vartheta - d)^T\mid\bm{\mathsf{Z_n}}\right\}\rvert + p \\=& \log\lvert Var_{\Pi_n}(\vartheta\mid\bm{\mathsf{Z_n}}) + \{E_{\Pi_n}(\vartheta \mid \bm{\mathsf{Z_n}}) - d\}\{E_{\Pi_n}(\vartheta \mid \bm{\mathsf{Z_n}}) - d\}^T\rvert + p.\end{align*}
By matrix determinant lemma, the above expression equals
\begin{align*}
&\log\lvert Var_{\Pi_n}(\vartheta\mid\bm{\mathsf{Z_n}})\rvert + \log\bigg[1 + \left\{E_{\Pi_n}(\vartheta\mid\bm{\mathsf{Z_n}}) - d\right\}^TVar_{\Pi_n}(\vartheta\mid\bm{\mathsf{Z_n}})^{-1}\left\{E_{\Pi_n}(\vartheta\mid\bm{\mathsf{Z_n}}) - d\right\}\bigg] + p\end{align*}
which is minimized by setting $d = E_{\Pi_n}\left(\vartheta\mid\bm{\mathsf{Z_n}}\right)$, which in turn means setting $\Sigma = Cov_{\Pi_n}(\vartheta\mid\bm{\mathsf{Z_n}})$. This rule achieves minimized posterior risk $E_{\Pi_n}(L\mid\bm{\mathsf{Z_n}})=\log\lvert Var_{\Pi_n}(\vartheta\mid\bm{\mathsf{Z_n}})\rvert+p$.

\subsection{Proof of Proposition~3.1}

The expected posterior loss is
\begin{align}
    &E_{\Pi_n}(L_{BI}\mid\bm{\mathsf{Z_n}}) = \log\lvert\Sigma\rvert + \label{eq:bil-risk}\\ &\qquad E_{\Pi_n}\left\{(\vartheta - d)^T\Omega\Sigma^{-1}(\vartheta - d)\mid \bm{\mathsf{Z_n}}\right\} + E_{\Pi_n}\left[\dfrac{1}{n}\sum_{i=1}^n \dot l_i(\vartheta)^T\{\Omega I_n(\vartheta)\}^{-1}\dot l_i(\vartheta)\mid\bm{\mathsf{Z_n}}\right]\nonumber
\end{align}

Similar to the proof of Proposition~2.1, the minimum of (\ref{eq:bil-risk}) with respect to $\Sigma$ with $\Omega$ fixed is achieved by setting
$\Sigma = E_{\Pi_n}\left\{(\vartheta - d)(\vartheta - d)^T \mid \bm{\mathsf{Z_n}}\right\}\Omega$ and $d = E_{\Pi_n}(\vartheta\mid\bm{\mathsf{Z_n}})$. Substituting $E_{\Pi_n}\left\{(\vartheta - d)(\vartheta -d)^T \mid  \bm{\mathsf{Z_n}}\right\}\Omega$ and $\hat d = E_{\Pi_n}(\vartheta\mid \bm{\mathsf{Z_n}})$ in Equation (\ref{eq:bil-risk}), we obtain that the minimal expected posterior loss with $\Omega$ fixed is
\begin{align}
    \label{eq:bil-loss-Omega}&\log \lvert Var_{\Pi_n}(\vartheta\mid \bm{\mathsf{Z_n}}) \rvert + p + \log \lvert\Omega\rvert + E_{\Pi_n}\left[\dfrac{1}{n}\sum_{i=1}^n \dot l_i(\vartheta)^T\{\Omega I_n(\vartheta)\}^{-1}\dot l_i(\vartheta) | \bm{\mathsf{Z_n}}\right]
\end{align}

Again using the same method as in the proof of Proposition~2.1, the minimum of (\ref{eq:bil-loss-Omega}) is achieved by setting
$\Omega = E\left\{\dfrac{1}{n}\sum_{i=1}^n \dot l_i(\vartheta)\dot l_i(\vartheta)^TI_n^{-1}(\vartheta) \mid \bm{\mathsf{Z_n}}\right\}$, which in turn gives the Bayes rule of $\Sigma$:
$$\hat\Sigma = Var_{\Pi_n}(\vartheta | \bm{\mathsf{Z_n}}) E_{\Pi_n}\left\{\dfrac{1}{n}\sum_{i=1}^n \dot l_i(\vartheta)\dot l_i(\vartheta)^TI_n^{-1}(\vartheta) \mid \bm{\mathsf{Z_n}}\right\}.$$

\subsection{Proof of Theorem~3.2}\label{append:proof-brse}
Let $I(\theta) = -E_0\left\{\frac{\partial^2}{\partial\theta^2}\log p_{\theta}(Z_1)\right\}$ be the Fisher information. By the Bernstein-von-Mises Theorem under model misspecification~\citep{kleijn2012bernstein}, $nVar_{\Pi_n}(\vartheta\mid\bm{\mathsf{Z_n}})\approx I_{\theta^*}^{-1}$ where $\theta^* = \argmin_\theta E_0\left\{\dfrac{p_0(\bm Z_i)}{p_\theta(\bm Z_i)}\right\}$ is the minimal Kullback-Leibler point. We see that, by Law of Large Numbers and Continuous Mapping Theorem, $$\dfrac{1}{n}\sum_{i=1}^n \dot l_i(\theta)\dot l_i(\theta)^TI_n(\theta)^{-1}\rightarrow_{P_0} E_0\{\dot l_i(\theta)\dot l_i(\theta)^T\}I(\theta)^{-1}$$ for each  $\theta\in\Theta$. The theorem below concludes that under regularity conditions in Section~D of the Supplementary Materials, we have

\begin{align*}
&E_{\Pi_n}\left\{\sum_{i=1}^n \dot l_i(\vartheta)\dot l_i(\vartheta)^TI_n(\vartheta)^{-1} \mid \bm{\mathsf{Z_n}}\right\}\overset{P_0}{\rightarrow}E_0\{\dot l_i(\theta^*)\dot l_i(\theta^*)^T\}I(\theta^*)^{-1}.\end{align*}

\begin{lemma}
\label{lemma:brcm}
Suppose $A_n(\theta)$ is a random function and measurable on $(\bm Z_1, \dots, \bm Z_n)$.  Suppose the conditions in \citet{kleijn2012bernstein} Theorem 2.3 and Corollary 4.2 hold. Further assume the following conditions hold:
\begin{enumerate}
    \item $A$ is continuous in a neighborhood of $\theta^*$ and $\lVert A(\theta^*)\rVert <\infty$.
    \item Uniform convergence of $A_n$ in a neighborhood of $\theta^*$: for any $\epsilon, \gamma >0$, there exists $\eta_0>0$ and positive integer $N$ such that for any  $n>N$, we have $$P_0(\sup_{\theta_1: \lvert  \theta_1 - \theta^* \rvert  < \eta_0}\lVert A_n(\theta_1) - A(\theta_1)\rVert >\epsilon)<\gamma.$$
    \item Asymptotic posterior uniform integrability:
    for any $\epsilon>0$, we have
\end{enumerate}
\begin{small}
\begin{align}
    &\lim_{M\rightarrow \infty}\limsup_{n\rightarrow\infty}P_0\bigg(E_{\Pi_n}\bigg\{ \lVert  A_n(\vartheta) \rVert  \mathbbm{1}\left(  \lVert  A_n(\vartheta) \rVert >M\right)  \mid \bm{\mathsf{Z_n}}\bigg\}>\epsilon\bigg)=0\label{eq:unif-intb}
    \end{align}
\end{small}
Then $E_{\Pi_n}\{A_n(\vartheta)\mid\bm{\mathsf{Z_n}}\}\rightarrow_{P_0} A(\theta^*)$.
\end{lemma}

\begin{proof}

\textbf{Step 1:}

We show that $\Pi_n(\lVert A_n(\vartheta) - A_n(\theta^*)\rVert>\epsilon \mid \bm{\mathsf{Z_n}})\rightarrow_{P_0}0$ as $n\rightarrow \infty$, for any $\epsilon>0$. In other words, for any $\epsilon$, $\gamma$, and $\xi > 0$, there exists a positive integer $N>0$, such that for any $n>N$, we have 

\begin{align}&P_0\left(\Pi_n\left( \lVert A_n(\vartheta) - A_n(\theta^*) \rVert  > \epsilon  \mid   \bm{\mathsf{Z_n}}\right)>\xi\right) < 2\gamma.\label{eq:post-convergence}\end{align}

Note that the additional conditions 1, 2 imply that for any $\epsilon, \gamma>0$, there exists $\eta_0>0$ and positive integer $N_1$ such that for any $n>N_1$, we have \begin{equation}\label{eq:unif_cont}P_0(\sup_{\theta_1: \lVert  \theta_1 - \theta^* \rVert < \eta_0}\lVert A_n(\theta_1) - A_n(\theta^*)\rVert>\epsilon)<\gamma.\end{equation}

By Corollary 4.2 in \citet{kleijn2012bernstein}, there exists a positive integer $N_2$ such that $$P_0\left(\Pi_n\left( \lVert \vartheta - \theta^* \rVert \geq \eta_0 \mid \bm{\mathsf{Z_n}}\right)>\xi\right)<\gamma$$ for any $n > N_2$.

Now that the left-hand-side of (\ref{eq:post-convergence}) is
\begin{small}
\begin{align*}
    &P_0\bigg(\Pi_n\bigg( \lVert  A_n(\vartheta) - A_n(\theta^*) \rVert  > \epsilon,  \lVert \vartheta - \theta^* \rVert<\eta_0  |  \bm{\mathsf{Z_n}}\bigg)>\xi\bigg) + \\&\qquad P_0\bigg(\Pi_n\bigg( \lVert  A_n(\vartheta) - A_n(\theta^*) \rVert  > \epsilon, \lVert \vartheta - \theta^* \rVert\geq \eta_0 \mid  \bm{\mathsf{Z_n}}\bigg)>\xi\bigg)\\
    \leq & P_0\bigg(\sup_{\theta: \lVert \theta - \theta^* \rVert  < \eta_0} \lVert  A_n(\theta) - A_n(\theta^*) \rVert  > \epsilon,   \Pi_n\left( \lVert  \vartheta - \theta^* \rVert <\eta_0 \mid  \bm{\mathsf{Z_n}}\right)>\xi\bigg) + P_0\left(\Pi_n\left( \lVert \vartheta - \theta^* \rVert\geq \eta_0 \mid  \bm{\mathsf{Z_n}}\right)>\xi\right)\\
    \leq & P_0\left(\sup_{\theta: \lVert \theta - \theta^* \rVert < \eta_0} \lVert  A_n(\theta) - A_n(\theta^*) \rVert  > \epsilon\right) +  P_0\left(\Pi_n\left(\lVert\vartheta - \theta^* \rVert\geq \eta_0 \mid  \bm{\mathsf{Z_n}}\right)>\xi\right)\\
    < &2\gamma.
\end{align*}
\end{small}

\textbf{Step 2:} 
We show that $$E_{\Pi_n}\{A_n(\vartheta)\mid\bm{\mathsf{Z_n}}\} = A_n(\theta^*) + O_{P_0}(1).$$

Note that the uniform continuity (\ref{eq:post-convergence}) also holds for functions $A_n(\cdot)\wedge M$ with $M>0$. The results of Step 1 also applies to $A_n(\cdot)\wedge M$:

\begin{align*} 
&\Pi_n(\lVert A_n(\vartheta)\wedge M - A_n(\theta^*)\wedge M\rVert >\epsilon \mid \bm{\mathsf{Z_n}}) \rightarrow_{P_0}0 
\end{align*}
for any $\epsilon > 0$.

It suffices to show that $$E_{\Pi_n}(\lVert A_n(\vartheta) - A_n(\theta^*)\rVert\mid\bm{\mathsf{Z_n}})\rightarrow_{P_0}0.$$

By triangular inequality we have
\begin{align*}
   & E_{\Pi_n}(\lVert A_n(\vartheta) - A_n(\theta^*)\rVert \mid\bm{\mathsf{Z_n}})\\ =& E_{\Pi_n}(\lVert A_n(\vartheta) - A_n(\vartheta)\wedge M\rVert \mid\bm{\mathsf{Z_n}}) + E_{\Pi_n}(\lVert A_n(\vartheta)\wedge M - A_n(\theta^*)\wedge M\rVert\mid\bm{\mathsf{Z_n}}) + \\& \lVert A_n(\theta^*)\wedge M - A_n(\theta^*)\rVert.
\end{align*}

Now on the right-hand side, the third term converges in $P_0$ to $\lVert A(\theta^*)\wedge M - A(\theta^*)\rVert$ as $n\rightarrow\infty$ and then converges in $P_0$ to $0$ by letting $M\rightarrow\infty$.

The first term satisfies
\begin{align*}
    &E_{\Pi_n}\left( \lVert  A_n(\vartheta) - A_n(\theta^*) \rVert \bm{\mathsf{Z_n}}\right) \leq  E_{\Pi_n}\left\{ \lVert  A_n(\theta) \rVert \mathbbm{1}\left( \lVert  A_n(\theta)\rVert>M\right) | \bm{\mathsf{Z_n}}\right\},
\end{align*}
which converges to zero by letting $n\rightarrow\infty$ followed by letting $M\rightarrow\infty$.

Finally, the second term is 
\begin{align*}
    &\quad E_{\Pi_n}\bigg( \lVert  A_n(\vartheta)\wedge M - A_n(\vartheta)\wedge M \rVert \mathbbm{1}\left( \lVert \theta - \theta^* \rVert \leq \eta_0\right) \mid \bm{\mathsf{Z_n}}\bigg) +\\
    &\qquad  E_{\Pi_n}\bigg( \lVert  A_n(\vartheta)\wedge M - A_n(\vartheta)\wedge M \rVert  \mathbbm{1}\left( \lVert \vartheta - \theta^* \rVert  \leq \eta_0\right)  \mid \bm{\mathsf{Z_n}}\bigg)\\
    &\leq  \sup_{\theta_1: \lVert\theta_1 - \theta^* \rVert < \eta_0} \lVert A_n(\theta_1)\wedge M - A_n(\theta^*)\wedge M \rVert   \Pi_n\left( \lVert \vartheta - \theta^* \rVert  < \eta_0 \mid  \bm{\mathsf{Z_n}}\right) +  M\cdot \Pi_n\left( \lVert \vartheta - \theta^* \rVert  < \eta_0 \mid \bm{\mathsf{Z_n}}\right)\\
    &\underset{P_0}{\rightarrow} 0 + 0 = 0
\end{align*}
by letting $n\rightarrow \infty$ followed by letting $M\rightarrow \infty$.
\end{proof}

\newpage
\section{Regularity Conditions for Theorem 3.2}
Let $\Theta$ be an open subset of $\mathbb R^p$ parametrizing statistical models $\{P_\theta:\theta\in\Theta\}$ with densities $p_\theta$. Suppose the observations $Z_1,\dots,Z_n$ form an i.i.d. sample from a distribution $P_0$ with density $p_0$ relative to a common dominating measure. Let $\theta^*\in\Theta$ be the minimal Kullback-Leibler point, i.e.,
$$-E_0\{\log p_{\theta^*}(Z_i)/\log p_0(Z_i)\}=\underset{\theta\in\Theta}{\inf}[-E_0\{\log p_{\theta}(Z_i)/\log p_0(Z_i)\}].$$ Assume the conditions in Kleijn et al. (2012) Theorem 2.3 and Corollary 4.2 hold:
\begin{condition}[Stochastic local asymptotic normality; Lemma 2.1 in \citet{kleijn2012bernstein}]\label{rc:lan} 
the function $\theta\mapsto \log p_\theta(Z_i)$ is differentiable at $\theta^*$ in $P_0$-probability with derivative $\dot l_{\theta^*}(Z_i)$ and
\begin{enumerate}[(i)]
    \item there is an open neighborhood $U$ of $\theta^*$ and a square-integrable function $m_{\theta^*}$ such that for all $\theta_1,\theta_2\in U$:
    $$\lvert\log\dfrac{p_{\theta_1}}{p_{\theta_2}}\rvert\leq m_{\theta^*}\lVert\theta_1 - \theta_2\rVert,\quad (P_0-a.s.).$$
    Additionally, we require $E_0\{p_\theta(Z_i)/p_{\theta^*}(Z_i)\}$ for all $\theta\in U$ and $E_0\{\exp(sm_{\theta^*})\}<\infty$ for some $s>0$.
    \item The Kullback-Leiber divergence with respect to $P_0$ has a second-order Taylor expansion around $\theta^*$:
    $$-E_0\left\{\log\dfrac{p_\theta}{p_{\theta^*}}(Z_i)\right\}=\dfrac{1}{2}(\theta - \theta^*)V_{\theta^*}(\theta - \theta^*)+o(\lVert\theta - \theta^*\rVert^2),\quad (\theta\rightarrow \theta^*),$$
    where $V_{\theta^*}$ is a positive-definite $d\times d$-matrix.
\end{enumerate}
    
\end{condition}

\begin{condition}[Conditions for convergence of posterior distribution; Theorem 3.1 of \citet{kleijn2012bernstein}]\label{rc:conv-post} The density function of the prior distribution $\pi(\cdot)$ is continuous and positive in a neighborhood of $\theta^*$. Furthermore, $E_0\{\dot l_{\theta^*}(Z_i)\dot l_{\theta^*}(Z_i)^T\}$ is invertible and for every $\epsilon>0$, there exists a sequence of tests $(\phi_n)$ such that
$$E_0\{\phi_n(Z_1,\dots, Z_n)\}\rightarrow 0,\qquad \underset{\{\theta:\lVert\theta - \theta^*\rVert\geq \epsilon\}}{Q_{\theta}\{1 - \phi_n(Z_1,\dots,Z_n)\}}\rightarrow 0.$$    
where $Q_\theta$ is a  measure defined by
$$Q_\theta(A)=E_0\left(\dfrac{p_\theta}{p_{\theta^*}}\mathbbm 1(A)\right)$$
for all $\theta\in\Theta$.
\end{condition}

\begin{condition}[Smoothness and positivity of prior density around $\theta^*$]\label{rc:smooth-prior}
    Let $V\subset\Theta$ be a measurable subset of $\Theta$. Assume that for all $\epsilon>0$, the prior distribution satisfies
    $$\Pi\left(\vartheta\in\Theta:-E_0\left\{\dfrac{\log p_\vartheta(Z_i)}{\log p_{\theta^*}(Z_i)}\right\}\leq \epsilon\right)>0.$$
\end{condition}

Conditions RC~\ref{rc:lan} roughly states that the statistical model $P_\theta$ is sufficiently ``smooth'' around $\theta^*$ and is the standard condition for the existence of a regular estimator~\cite[\S 7]{van2000asymptotic}. Condition RC~\ref{rc:conv-post} requires that with enough data, the statistical model is identifiable at $\theta^*$. These two conditions are  and are sufficient conditions for the posterior distribution to converge to $\theta^*$ at rate $1/\sqrt n$ (Theorem 3.1 of \citet{kleijn2012bernstein}). RC~\ref{rc:smooth-prior} is a stronger condition that RC~\ref{rc:smooth-prior} is stronger that RC~\ref{rc:conv-post} in that the prior density is sufficiently smooth positive in a neighborhood around $\theta^*$. By Corollary 4.2, conditions~RC~\ref{rc:conv-post} and RC~\ref{rc:smooth-prior} are sufficient conditions that for any $\eta > 0$ , $$\Pi_n\left(\lVert \vartheta-\theta_*\rVert<\eta \mid \bm{\mathsf Z_n}\right)\rightarrow 1,\qquad P_0-a.s.$$ This is a form of ``strong consistency'' and a result also referred to as ``Schwartz consistency'' theorem.

Finally, 
\begin{condition}[Continuity, uniform convergence and posterior uniform integrability]\label{rc:lemma-s1} Let $$A_n(\theta)=\dfrac{1}{n}\sum_{i=1}^n \dot l_i(\theta)\dot l_i(\theta)^TI_n(\theta)^{-1}$$ and $$A(\theta) = E_0\{\dot l_i(\theta)\dot l_i(\theta)^T\}I(\theta)^{-1}.$$
The matrix-valued functions $A_n(\theta)$ and $A(\theta)$ satisfy the conditions 1-3 of Lemma~\ref{lemma:brcm} in Section A.3.
\end{condition}

In RC~\ref{rc:lemma-s1}, we note that the uniform convergence is implied by Strong Law of Large Numbers, and that the asymptotic uniform integrability roughly requires that the posterior tail probability of $\lVert A_n(\vartheta)\rVert$ is bounded almost surely.

\newpage
\section{Balanced loss functions for quasi-likelihood regression}\label{append:quasi-likelihood}

A slightly simpler version of the Balanced Inference Loss function provides a Bayesian analog for quasi-likelihood variance estimates, in regression models. Consider the same setup as in Section 4.2, with the additional assumption that $E(Y_i\mid\bm X_i) = \alpha V_i$ for some $\alpha > 0$. We propose the following balanced inference loss function

\begin{align}&L_{BI}(\beta,d, \Sigma, \omega) =\log\mid\Sigma\mid + (\beta - d)^T\omega\Sigma^{-1}(\beta - d) + \dfrac{p}{n\omega}\sum_{i=1}^n \frac{(Y_i - \mu_i)^2}{V_i},\label{eq:bil-quasilikelihood}\end{align}

where $\omega>0$ is a univariate scaling factor. In contrast to the general loss function for GLMs of Equation (3), in Equation (\ref{eq:bil-quasilikelihood}) the correction term accounting for lack of model fit  is proportional to the average Pearson residual $\dfrac{1}{n}\sum_{i=1}^n (Y_i - \mu_i)^2/V_i$. The Bayes rules for the loss function (\ref{eq:bil-quasilikelihood}) set $d$ to be the usual posterior mean of $\beta$, $\omega$ to be the posterior mean of average Pearson residual 
\begin{equation}\hat\omega =E_{\Pi_n}\left\{\dfrac{1}{n}\sum_{i=1}^n \dfrac{(Y_i - \mu_i)^2}{V_i}\mid \bm Y, \bm{\mathsf{X}}\right\},\label{eq:scaling-quasi}\end{equation}
and $\Sigma$ to be the posterior variance scaled by $\hat\omega_n$, i.e. $\hat\Sigma_n = \hat\omega_n \cdot \mathrm{Var}_{\Pi_n}\left(\beta\mid\bm Y, \bm{\mathsf{X}}\right)$. We see that $\hat\Sigma_n$ gives a Bayesian analog of the variance estimate for the frequentist quasi-likelihood method~\citep{wedderburn1974quasi}.

\newpage
\section{Proof of Results in Section 4}
\subsection{Results in Section 4.1}
Assume $Y_1,\dots,Y_n\sim N(\theta, \sigma_0^2)$, but a wrong statistical model was used which assumes $Y_1,\dots, Y_n\sim N(\theta, 1)$.

We have
\begin{align*}
p_\theta(Y_i) &= -\dfrac{1}{\sqrt{2\pi}}\exp\{-\dfrac{1}{2}(Y_i-\theta)^2\},\\
l_i(\theta) &= \dfrac{\partial}{\partial\theta}\log p_\theta(Y_i)= Y_i - \theta,\\
I_n(\theta) &= -\dfrac{1}{n}\sum_{i=1}^n \dfrac{\partial}{\partial\theta}l_i(\theta)=1.
\end{align*}
By Proposition 3.1, we have
\begin{align*}
    \hat\omega &= E_{\Pi_n}\left\{\dfrac{1}{n}\sum_{i=1}^n (Y_i - \vartheta)^2\mid Y_1,\dots, Y_n\right\}\\
    &= \dfrac{1}{n}\sum_{i=1}^n (Y_i -\bar Y)^2 + E\{(\vartheta - \bar Y)^2\mid Y_1,\dots, Y_n\}
\end{align*}

It is straightforward to derive that under the prior distribution $\vartheta \sim N(\mu,\eta^2)$ and the wrongly assumed model $Y_1,\dots, Y_n\mid \vartheta = \theta \sim N(\theta, 1)$, the posterior distribution is
$$N\left(\dfrac{\mu + n\eta^2 \bar Y}{1 + n\eta^2}, \dfrac{\eta^2}{1 + n\eta^2}\right).$$

Hence
\begin{align*}
    E_{\Pi_n}(\vartheta\mid Y_1,\dots, Y_n) &= \dfrac{\mu + n\eta^2 \bar Y}{1 + n\eta^2},\quad Var_{\Pi_n}(\vartheta\mid Y_1,\dots, Y_n) = \dfrac{\eta^2}{1 + n\eta^2},\\
     E_{\Pi_n}\{(\vartheta-\bar Y)^2\mid Y_1,\dots, Y_n\} &= \left(\dfrac{\mu + n\eta^2 \bar Y}{1 + n\eta^2}-\bar Y\right)^2 + \dfrac{\eta^2}{1 + n\eta^2}
\end{align*}
We concluded that
\begin{align*}
    \hat d &= \dfrac{\mu + n\eta^2 \bar Y}{1 + n\eta^2}\\
    \hat \omega &= \dfrac{1}{n}\sum_{i=1}^n (Y_i - \bar Y)^2 + \left(\dfrac{\mu-\bar Y}{1 + n\eta^2}\right)^2 + \dfrac{\eta^2}{1 + n\eta^2}\\
    \hat\sigma^2 &= \dfrac{\eta^2}{n(1 + n\eta^2)}\sum_{i=1}^n (Y_i - \bar Y)^2 + \dfrac{\eta^2(\mu-\bar Y)^2}{(1 + n\eta^2)^3}+ \left(\dfrac{\eta^2}{1 + n\eta^2}\right)^2.
\end{align*}

\subsection{Results in Section 4.2}
It is easy to derive that the score function and scaled empirical Fisher information are
\begin{align*}
    \dot l_i(\beta) &= \dfrac{\partial}{\partial\beta}\log p_{\theta,\alpha}(y_i)\\
    &= \dfrac{\partial\mu_i}{\partial\beta}\dfrac{\partial\theta_i}{\partial\mu_i}\dfrac{\partial}{\partial\theta_i}\left\{\dfrac{Y_i\theta_i - b(\theta_i)}{\alpha}\right\}\\
    &= \dfrac{\partial\mu_i}{\partial\beta}\dfrac{1}{\partial\mu_i/\partial\theta_i}\dfrac{\partial}{\partial\theta_i}\left\{\dfrac{Y_i\theta_i - b(\theta_i)}{\alpha}\right\}\\
    &= \dfrac{\partial\mu_i}{\partial\beta}\dfrac{Y_i - \partial b(\theta_i)/\partial\theta_i}{\alpha (\partial\mu_i / \partial\theta_i)}\\
    &= \dfrac{\partial\mu_i}{\partial\beta}\dfrac{Y_i - \mu_i}{\alpha V_i},\\
    I_n(\beta) &= -\dfrac{1}{n}\sum_{i=1}^n\dfrac{\partial}{\partial\beta}\dot l(\beta)\\
    &= -\dfrac{1}{n}\sum_{i=1}^n \bigg[\dfrac{\partial^2\mu_i}{\partial\beta\partial\beta^T}\left(\dfrac{Y_i - \mu_i}{\alpha V_i}\right) - \dfrac{1}{\alpha V_i}\dfrac{\partial\mu_i}{\partial\beta}\left(\dfrac{\partial\mu_i}{\partial\beta}\right)^T - \dfrac{\partial\mu_i}{\partial\beta}\left(\dfrac{Y_i - \mu_i}{\alpha V_i^2}\right)\left(\dfrac{\partial V_i}{\partial\beta}\right)^T\bigg]
\end{align*}
where $$V_i = \dfrac{\partial \mu_i}{\partial\theta_i}\mid_{\theta_i = (b')^{-1}\{g^{-1}(x_i^T\beta)\}}.$$

If a canonical link is used such that $g(\mu_i)=\theta_i = X_i^T\beta$, then
\begin{align*}
    \dfrac{\partial\mu_i}{\partial\beta} &= \dfrac{\partial\mu_i}{\partial\theta_i}\times \dfrac{\partial\theta_i}{\partial\beta} = V_iX_i,\quad \dfrac{\partial^2\mu_i}{\partial\beta\partial\beta^T} = X_i\left(\dfrac{\partial V_i}{\partial\beta}\right)^T.
\end{align*}
Hence
\begin{align*}
    \dot l_i(\beta) &= X_i(Y_i - \mu_i)/\alpha,\\
    I_n(\beta) &= -\dfrac{1}{n}\sum_{i=1}^n \bigg[X_i\left(\dfrac{\partial V_i}{\partial\beta}\right)^T\left(\dfrac{Y_i - \mu_i}{\alpha V_i}\right) - \dfrac{1}{\alpha V_i}X_iX_i^TV_i^2 - X_iV_i\left(\dfrac{\partial V_i}{\partial\beta}\right)^T\left(\dfrac{Y_i - \mu_i}{\alpha V_i^2}\right)\bigg]\\
    &= \dfrac{1}{n}\sum_{i=1}^n X_iX_i^TV_i/\alpha.
\end{align*}

\subsubsection*{Example: Linear regression} 
In linear regression, we have $\mu_i=\theta_i = X_i^T\beta$, $V_i=1$ and $\alpha=\sigma^2$. The score function and the scaled empirical Fisher information therefore are
\begin{align*}
    \dot l_i(\beta) &= X_i(Y_i - \mu_i)/\sigma^2,\qquad
    I_n(\beta) = \dfrac{1}{n}\sum_{i=1}^n X_iX_i^T/\sigma^2.
\end{align*}

\subsubsection*{Example: Poisson regression}
In Poisson regression, we have
$\mu_i = V_i = \exp(X_i^T\beta)$ and $\alpha=1$. The score function and scaled empirical Fisher information therefore are

\begin{align*}
    \dot l_i(\beta) &= X_i(Y_i - \mu_i),\qquad
    I_n(\beta) = \dfrac{1}{n}\sum_{i=1}^n X_iX_i^T\exp(X_i^T\beta).
\end{align*}

\subsection{Results in Section 4.3}

The assumed model is
\begin{align*}p_{\beta}(\widetilde T_i, X_i, \Delta_i) &= \left[\exp(X_i^T\beta)\exp\{-\exp(X_i^T\beta)\widetilde T_i\}\right]^{\Delta_i} \left[\exp\{-\exp(X_i^T\beta)\widetilde T_i\}\right]^{1-\Delta_i}\\
&= \exp\left\{\Delta_iX_i^T\beta - \exp(X_i^T\beta)\widetilde T_i\right\}.\end{align*}

The score function and (scaled) empirical Fisher information function are
\begin{align*}
    \dot l_i(\beta) &= \{\Delta_i - \exp(X_i^T\beta)\widetilde T_i\}X_i\\
    I_n(\beta) &= \dfrac{1}{n}\sum_{i=1}^nX_iX_i^T\exp(X_i^T\beta)\widetilde T_i
\end{align*}
Therefore, the balance inference loss function is

\begin{align*}&L_{BI}(\beta, d, \Sigma, \Omega) = \log\mid\Sigma\mid + (\beta - d)^T\Omega\Sigma^{-1}(\beta - d) + \\&\qquad \sum_{i=1}^n \{\Delta_i - \widetilde T_i\exp(\bm X_i^T\beta)\}^2  \bm X_i^T\left[\Omega\left\{\sum_{j=1}^n \widetilde T_j\exp(\bm X_j^T\beta)\bm X_j\bm X_j^T\right\}\right]^{-1}\bm X_i.\end{align*}

\newpage
\section{Simulation study for fixed-design linear regression}

We use the same simulation setting for the random-design linear regression in Section 4.2, except that the regressors are $X_i=(1, U_i)$, where $U_i$'s are fixed and evenly spaced between $0$ and $3$. Table~\ref{tab:brse-fixed-lm} and Figure~\ref{fig:int_cov} respectively show the comparison of different standard error estimates and the (frequentist) coverage probabilities over the minimal Kullback-Leibler point of different uncertainty intervals.

We see that inferences based on posterior standard deviation, robust standard error and Bayesian robust standard error are all conservative except when the model is correctly specified. With very small sample sizes ($n=5$), the confidence intervals with robust standard error have under-coverage, while the Bayesian robust confidence intervals remain conservative.

\begin{table}[!htbp]
\caption{\label{tab:brse-fixed-lm} Comparison between posterior standard deviation and the Bayesian robust standard error estimate to the true standard error of the Bayes point estimate $\hat d$ in fixed-design linear regression. Ave($\hat d$): Monte Carlo average of $\hat d$; $SE(\hat d)$: Monte Carlo standard deviation of $\hat d$; Ave(Post.SD): average posterior standard deviation; Ave($\widehat{BRSE}$): Monte Carlo average of Bayesian robust standard error. Gray rows indicate where the model is correctly specified.}
\centering
\vskip 1em

    \begin{tabular}{rr|rrrrrrr}
  \hline
n & $a$ & Ave($\hat d$) & $SE(\hat d)$ &  Ave(Post.SD) & $Ave(\widehat{BRSE})$ \\ 
  \hline
50 & -2 & -5.001 & 0.160 & 0.288 & 0.339 \\ 
  50 & -1 & -1.994 & 0.163 & 0.202 & 0.220 \\ 
  \rowcolor{Gray} 50 & 0 & 1.005 & 0.161 & 0.163 & 0.162 \\ 
  50 & 1 & 4.007 & 0.162 & 0.201 & 0.219 \\ 
  50 & 2 & 6.993 & 0.161 & 0.287 & 0.336 \\ 
  100 & -2 & -4.996 & 0.120 & 0.203 & 0.237 \\ 
  100 & -1 & -2.002 & 0.116 & 0.142 & 0.154 \\ 
  \rowcolor{Gray}  100 & 0 & 1.006 & 0.114 & 0.115 & 0.115 \\ 
  100 & 1 & 4.005 & 0.117 & 0.142 & 0.154 \\ 
  100 & 2 & 7.000 & 0.112 & 0.202 & 0.236 \\ 
   \hline
\end{tabular}
\end{table}

\begin{figure}[!htbp]
    \centering

    \includegraphics[width = 1\linewidth]{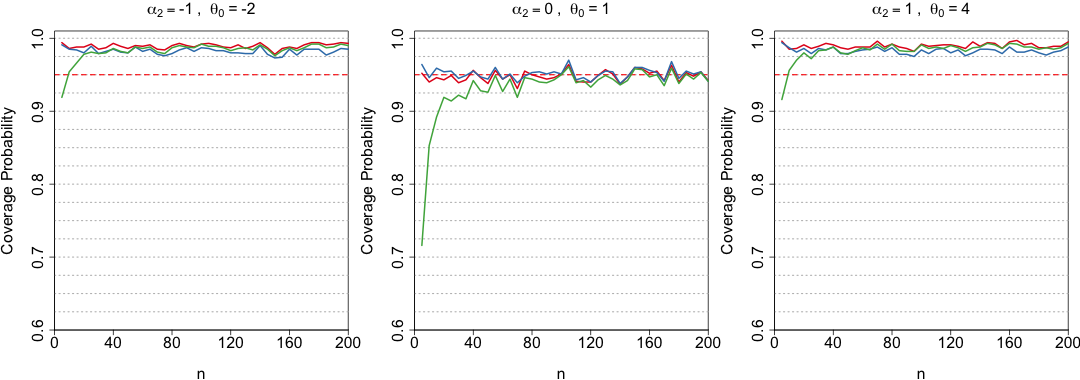}
    \caption{Coverage probabilities of 95\% credible interval (red), frequentist (blue) and Bayesian (green) robust confidence intervals for fixed-regressor linear regression. The middle column shows the results with  correctly-specified models. $\theta_0$ denotes the minimal Kullback-Leibler point.}
    \label{fig:int_cov}
\end{figure}

\newpage
\bibliographystyle{biom}
\bibliography{references}